\documentclass[11pt,a4paper,reqno]{amsart}  
\usepackage{mathtools, amsmath, amsthm, amssymb, xcolor, hyperref, graphicx,multirow,array}
\newtheorem{lemma}{Lemma}

\usepackage{chngcntr}
\makeatletter
\def\hlinew#1{
	\noalign{\ifnum0=`}\fi\hrule \@height #1 \futurelet
	\reserved@a\@xhline}
\makeatother

\renewcommand{\arraystretch}{1.27}

\allowdisplaybreaks[1]

\theoremstyle{plain}
\newtheorem{thm}{Theorem} 

\newtheorem{cor}{Corollary}
\newtheorem{prop}{Proposition}
\newtheorem{conj}{Conjecture}
\theoremstyle{definition}
\newtheorem{remark}{Remark}

\newtheorem*{thm*}{Theorem}

\usepackage{graphicx,color}

\newcommand{\ZZ}{\mathbb{Z}}

\newcommand{\Mgnbar}{\overline{\mathcal{M}}_{g, n}}

\newcommand{\Mgtwobar}{\overline{\mathcal{M}}_{g, 2}}
\newcommand{\p}{\partial}
\newcommand{\e}{\epsilon}
\newcommand{\X}{\mathcal{X}}

\def\beq{\begin{equation}}
	\def\eeq{\end{equation}}
\def\be{\begin{equation}}
	\def\ee{\end{equation}}
\def\bes{\begin{equation*}}
	\def\ees{\end{equation*}}
\def\thin{\hspace 0.5pt}
\newcommand{\nn}{\nonumber}
\def\={\;=\;}   \def\+{\,+\,} \def\m{\,-\,}  \def\:{\;:=\;}
    \def\a{\alpha} \def\dd{\mathbf d}   \def\g{\gamma}

\def\thin{\hskip 1 pt}

\begin{document}
\title{On uniform large genus asymptotics of Witten's intersection numbers}
\author{Jindong Guo, Di Yang, Don Zagier}
\address{Jindong Guo, School of Mathematical Sciences, University of Science and Technology of China, 230026 Hefei, P.R.~China}
\email{guojindong@mail.ustc.edu.cn}
\address{Di Yang, School of Mathematical Sciences, University of Science and Technology of China, 230026 Hefei, P.R.~China}
\email{diyang@ustc.edu.cn}
\address{Don Zagier, Max Planck Institute for Mathematics, 53111 Bonn, 
	Germany, and International Centre for Theoretical Physics, Trieste, Italy}
\email{dbz@mpim-bonn.mpg.de}
\begin{abstract}
Following ideas from~\cite{GNYZ}, we give a uniform large genus 
asymptotics for primitive psi-class intersection numbers on the moduli space of stable algebraic curves, 
and extend this result including insertions of zeros in a certain uniform way. 
Application to a particular formal solution of the Painlev\'e I equation is given. 
We also use a method from~\cite{GNYZ} to give a new proof of the polynomiality conjecture on large genus asymptotic expansions 
of psi-class intersection numbers.
\end{abstract}
\maketitle
\tableofcontents
\section{Introduction and statements of the main results}
Let $\overline{\mathcal{M}}_{g,n}$ be the Deligne--Mumford moduli space 
of stable algebraic curves of genus $g$ with $n$ distinct marked points.
In~\cite{Witten} Witten proposed a striking conjecture 
on the relationship between topology of $\overline{\mathcal{M}}_{g,n}$ and the celebrated Korteweg--de Vries (KdV) 
integrable hierarchy. To be precise, let  
\beq\label{ZWK}
Z({\bf t};\e)\:\exp \biggl(\,\sum_{g,n} \frac{\e^{2g-2}}{n!} 
\sum_{d_1,\dots,d_n} t_{d_1} \cdots t_{d_n} \int_{\Mgnbar} \psi_1^{d_1} \cdots \psi_n^{d_n} \biggr)\,,
\eeq
be the {\it partition function of psi-class intersection numbers}, 
where ${\bf t}=(t_0,t_1,t_2,\dots)$ is an infinite vector of indeterminates, $\e$ is an indeterminate often called the string coupling constant, and 
$\psi_j$ denotes the first Chern class of the $j$th tautological line bundle on $\overline{\mathcal{M}}_{g,n}$.
Here and below, $\sum_i:=\sum_{i\ge0}$.
Define
\beq\label{defuZ}
u\=u({\bf t};\e)\:\e^2\, \frac{\p^2 \log Z({\bf t};\e)}{\p t_0^2}\,.
\eeq
Then Witten's conjecture claims that $u$ satisfies the KdV hierarchy:
\beq\label{KdV-hie}
\frac{\p u}{\p t_k} \= \frac{1}{(2k+1)!!}\, \Bigl[\bigl(L^{\frac{2k+1}2}\bigr)_+\,,\, L\Bigr]\,, \quad k\ge0\,,
\eeq
where $L:=\e^2 \p_{t_0}^2 + 2u$ is the Lax operator (cf.~e.g.~\cite{Dickey}). 
The $k=1$ equation in~\eqref{KdV-hie} reads 
\beq
\frac{\p u}{\p t_1} \= u \frac{\p u}{\p t_0} + \frac{\e^2}{12} \frac{\p^3 u}{\p t_0^3}\,,
\eeq
which is the celebrated KdV equation. Another way of stating Witten's conjecture is that $Z({\bf t};\e)$ is a tau-function 
for the KdV hierarchy; see e.g.~\cite{BDY,DJKM,Dickey,DYZ,DZ} for the notion of a KdV tau-function.

We will refer to the psi-class intersection numbers 
\begin{align}\label{WKnumbers}
\int_{\overline{\mathcal{M}}_{g,n}} \psi_1^{d_1} \cdots \psi_n^{d_n}\,,\qquad g,n,d_1,\dots,d_n\ge0\,,
\end{align}
as {\it Witten's intersection numbers}. 
They play important roles in studying 
 Weil--Petersson volumes of $\overline{\mathcal{M}}_{g,n}$~\cite{KMZ, LX09, ManZ, MZ, MS} and 
Masur--Veech volumes of moduli space of quadratic 
differentials~\cite{Agg, DGZZ22}. They are called {\it primitive} if $d_1\ge2$, $\dots$, $d_n\ge2$.
By the degree-dimension matching we know that the numbers~\eqref{WKnumbers} vanish unless 
\begin{align}\label{ddc}
d_1+\cdots+d_n\=3g-3+n\,.
\end{align} 

Witten's conjecture was first proved by Kontsevich~\cite{Kontsevich}. This conjecture also 
can~\cite{DVV} be equivalently stated that 
the partition function $Z=Z(\mathbf{t};\epsilon)$
satisfies the Virasoro constraints: 
\begin{align}\label{Virasoroconstra}
L_m (Z) \= 0\,, \quad \forall\, m\geq -1\,,
\end{align}
where $L_m$ are linear operators defined by
\begin{align}
&L_m \= -\frac{(2m+3)!!}{2^{m+1}} \frac{\p}{\p t_{m+1}} + \sum_{d} \frac{(2d+2m+1)!!}{(2d-1)!! \, 2^{m+1}} 
\,t_d\, \frac{\p }{\p t_{d+m}} \nn\\
& + \frac{\e^2} 2 \sum_{d=0}^{m-1} \frac{(2d+1)!!(2m-2d-1)!!}{2^{m+1}}\frac{\p^2}{\p t_d \p t_{m-1-d}} 
+ \frac{t_0^2}{2\e^2}\,\delta_{m,-1} + \frac1{16}\, \delta_{m,0}\,, \label{vira}
\end{align}
which satisfy the Virasoro commutation relations:
\beq
[L_{m_1},L_{m_2}]\=(m_1-m_2) \, L_{m_1+m_2}\,, \quad m_1,m_2\ge-1\,.
\eeq
It was shown in~\cite{DVV} that 
the Virasoro constraints~\eqref{Virasoroconstra} are equivalent to the following recursive relation:
\begin{align}
&U(\dd)\=\sum_{j=2}^n(2d_j+1) \, U(d_2,\dots,d_j+d_1-1,\dots,d_n)\nn\\
& +\frac12\sum_{\substack{a, \thin b \\ a+b=d_1-2}} \biggl( U(a,b,d_2,\dots,d_n) +\sum_{I\sqcup J=\{2,\dots,n\}}
U(a,\dd_I) \,
U(b,\dd_J) \biggr)\,, \label{DVV}
\end{align}
where $\dd:=(d_1,\dots,d_n)\in(\ZZ_{\ge0})^n$ satisfying $g(\dd):= 1+\frac{1}{3}\sum_{j=1}^{n}(d_j-1)\in\ZZ_{\ge0}$ and
 $2g(\dd)-2+n\ge2$, and $U(\dd)$ is defined by 
\begin{align}
U(\dd)\:\prod_{j=1}^{n}(2d_j+1)!!\,\int_{\overline{\mathcal{M}}_{g(\dd),n}} \psi_1^{d_1} \cdots \psi_n^{d_{n}}\,. \label{defUnormal}
\end{align}
Note that the integral on the right-hand side of~\eqref{defUnormal} is understood as~0 if $g(\dd)$ is not a non-negative integer.
  The relation~\eqref{DVV} is now widely known as 
 the {\it Dijkgraaf--Verlinde--Verlinde (DVV) relation}.

In~\cite{BDY} 
Bertola, Dubrovin and the second author of the present paper 
derived the following explicit formula of generating series of $n$-point Witten's intersection numbers:
\begin{align} 
\sum_{d_1,\dots,d_n}\frac{U(\dd)}{\prod_{i=1}^n\lambda_i^{d_i+1}} =
\left\{\begin{array}{ll}
\sum_{g\geq1}\frac{(6g-3)!!}{24^g \, g! \, \lambda^{3g-1}} \,, & n=1\,,\\
& \\
- \frac1{n} 
\sum_{\sigma\in S_n}\frac{{\rm tr} \, \prod_{i=1}^n M(\lambda_{\sigma(i)})}{\prod_{i=1}^{n}(\lambda_{\sigma(i)}-\lambda_{\sigma(i+1)})} - \frac{\delta_{n,2}(\lambda_1+\lambda_2)}{(\lambda_1-\lambda_2)^2}\,,  & n\ge2\,,
	\end{array}\right. \label{fnmr}
\end{align}
where 
$S_n$ denotes the symmetric group, for an element $\sigma\in S_n$ the notation $\sigma(n+1)$ is understood as $\sigma(1)$, and 
\begin{equation}\label{defM1002}
	M(\lambda) \: 
	\frac{1}{2}
	\begin{pmatrix}
		-\sum_{g=1}^{\infty}\frac{(6g-5)!!}{24^{g-1}(g-1)!}\lambda^{-3g+2} & -2\sum_{g=0}^{\infty}\frac{(6g-1)!!}{24^{g}g!}\lambda^{-3g} \\ 2\sum_{g=0}^{\infty}\frac{6g+1}{6g-1}\frac{(6g-1)!!}{24^{g}g!}\lambda^{-3g+1} & \sum_{g=1}^{\infty}\frac{(6g-5)!!}{24^{g-1}(g-1)!}\lambda^{-3g+2}
	\end{pmatrix}\,.
\end{equation}
The $n=1$ case of~\eqref{fnmr} is well known (see e.g.~\cite{Witten}). Following~\cite{GNYZ}, an explicit formula for 
Witten's intersection numbers will be deduced from~\eqref{fnmr} in Section~\ref{secgeneral}.
	
In~\cite{LX} Liu--Xu obtained 
the large genus asymptotics for Witten's intersection numbers with fixed~$n$ 
and fixed $d_1,\dots,d_{n-1}$; the requirement that $n$ should be fixed is essential in~\cite{LX}, meaning that 
it cannot be improved if one uses the normalizations considered in~\cite{LX}.
In~\cite{DGZZ20,DGZZ22} 
Delecroix--Goujard--Zograf--Zorich (DGZZ) introduced a remarkable normalization of Witten's intersection numbers
\beq\label{defGd}
G(\dd) \: \frac{24^{g(\dd)}\,g(\dd)!\,\prod_{j=1}^{n}(2d_j+1)!!}{(6g(\dd)+2n-5)!!} \, 
\int_{\overline{\mathcal{M}}_{g(\dd), \thin n}} \psi_1^{d_1} \cdots \psi_n^{d_n}\,, \quad \dd\in(\mathbb{Z}_{\ge0})^n\,,
\eeq
and proposed the following important conjecture: for any $C>0$,
\begin{equation}\label{DGZZconj}
\lim\limits_{g\to+\infty}\,\max_{1\le n \le C\log g}\, \max_{\substack{\dd\in(\ZZ_{\ge0})^n \\ |\dd|=3g-3+n}} 
	\bigl|G(\dd)-1\bigr|=0\,,
\end{equation}
or stated differently, the normalized intersection numbers $G(\dd)$ tend to~1 uniformly if $n=O(\log g)$. 
Here $|\dd|:=d_1+\dots+d_n$.

The DGZZ conjecture was first proved by Aggarwal~\cite{Agg}, and later independently by two of the authors 
of the present paper~\cite{GY} by a different method. Actually, Aggarwal~\cite[Theorem~1.5]{Agg} proved 
the following stronger version of the DGZZ conjecture:
\beq\label{asyAgg}
\lim_{\epsilon\to0}\biggl(\lim_{g\to\infty} \max_{n<\epsilon \sqrt{g}}
\max_{\substack{\dd\in(\ZZ_{\ge0})^n\\ |\dd|=3g-3+n}} \Bigl|G(\dd)-1\Bigr|\biggr) \= 0\,.
\eeq
In other words, $G(\dd)$ tends to~1 uniformly for $n=o(\sqrt{g})$. According to~\cite{Agg,DGZZ20},  
$o(\sqrt{g})$ cannot be replaced by $O(\sqrt{g})$ in this statement.

Recently, Norbury and the three authors of the present paper~\cite{GNYZ} discovered and 
proved a completely uniform large genus asymptotic formula for 
the Br\'ezin--Gross--Witten (BGW) numbers (for the meaning of BGW numbers see~e.g.~\cite{GNYZ}). 
In this paper, we aim to find and prove the corresponding result for Witten's intersection numbers.

Following the idea of~\cite{GNYZ}, introduce a new normalization by
\begin{align}\label{defCg}
C(\dd)\:\frac{2^{2g(\dd)}\, \prod_{j=1}^n (2d_j+1)!!}{3^{2 g(\dd)-2+n}\,(2g(\dd)-3+n)!} \, \int_{\overline{\mathcal{M}}_{g(\dd),\thin n}} \psi_1^{d_1} \cdots \psi_n^{d_n}\,,\quad \dd\in(\mathbb{Z}_{\ge0})^n\,.
\end{align}
The numbers $G(\dd)$ and $C(\dd)$ are related by
\beq\label{CGrelation}
 C(0^{n-1},3g(\dd)-3+n) \, G(\dd) \= C(\dd)\,.
\eeq
From the $n=1$ case of~\eqref{fnmr} we know that $C(3g-2)$ has the explicit expression:
\begin{align}\label{onepoint}
C(3g-2)\=\frac{3\, (6g-3)!!}{54^g\, g!\, (2g-2)!}\,,\quad g\ge1\,.
\end{align}
Explicit values of the numbers~$C(\dd)$ with $g(\dd)=2,3$ are given in Table~\ref{tablenormalizednumbers}.
\begin{table}[phbt]\label{tablenormalizednumbers}
\renewcommand\arraystretch{1.2}
	\begin{center}
		\begin{tabular}{!{\vrule width 1.5pt}l|c|l|r!{\vrule width 1.5pt}}
			\hlinew{1.5pt}
			\multicolumn{4}{!{\vrule width 1.5pt}c!{\vrule width 1.5pt}}{$g=2,\quad D=3888$}\\
			\hlinew{1.5pt}
			$(4)$&$\frac{35}{144}$&0.24306&945\\
			\hline
			$(2,3)$&$\frac{1015}{3888}$&0.26106&1015\\
			\hline
			$(2,2,2)$&$\frac{175}{648}$&0.27006&1050\\
			\hlinew{1.5pt}
			\multicolumn{4}{!{\vrule width 1.5pt}c!{\vrule width 1.5pt}}{$g=3,\quad D=7558272$}\\
			\hlinew{1.5pt}
			$(7)$&$\frac{25025}{93312}$&0.26819&2027025\\
			\hline
			$(2,6)$&$\frac{77077}{279936}$&0.27534&2081079\\
			$(3,5)$&$\frac{38731}{139968}$&0.27671&2091474\\
			$(4,4)$&$\frac{4249}{15552}$&0.27321&2065014\\
			\hline
			$(2,2,5)$&$\frac{6545}{23328}$&0.28056&2120580\\
			$(2,3,4)$&$\frac{39235}{139968}$&0.28031&2118690\\
			$(3,3,3)$&$\frac{714175}{2519424}$&0.28347&2142525\\
			\hline
			$(2,2,2,4)$&$\frac{6625}{23328}$&0.28399&2146500\\
			$(2,2,3,3)$&$\frac{179375}{629856}$&0.28479&2152500\\
			\hline
			$(2,2,2,2,3)$&$\frac{120625}{419904}$&0.28727&2171250\\
			\hline
			$(2,2,2,2,2,2)$&$\frac{546875}{1889568}$&0.28942&2187500\\\hlinew{1.5pt}
		\end{tabular}\vskip 3pt
		\caption{Some normalized Witten's intersection numbers $C(\dd)$}
	\end{center}
\end{table}

Similar to the conjectural nesting property for BGW numbers~\cite[Conjecture~1]{GNYZ}, we observe
from Table~\ref{tablenormalizednumbers} and further numerical data that the primitive normalized Witten's intersection numbers $C(\dd)$
with a fixed genus $g\leq13$ lie between the values $C(3g-2)$ and $C(2^{3g-3})$.
For example, for $g=13$ the numbers $C(\dd)$ with $d_1,\dots,d_n\ge2$ lie between 
$C(37) =  0.30674776\cdots$
and
$C(2^{36}) = 0.31263595\cdots$. 
We formulate the following
\begin{conj}\label{conjnesting}
	For $g\ge2$, $n\ge1$ and for $d_1,\dots,d_n\ge1$ satisfying $|\dd|=3g-3+n$, 
	\begin{align}\label{ordering}
		C(3g-2)\;\leq\; C(\dd)\;\leq\, C(2^{3g-3})\,.
	\end{align}
\end{conj}

\begin{remark}
	The inequality~\eqref{ordering} does not hold if $d_j$ is allowed to be~0. E.g., for $g=2$,
	\begin{align}
		&C(0^6,10)\=\frac{1616615}{6718464}\;\approx\; 0.24062 \;<\;0.24306\;\approx\; \frac{35}{144}\=C(4)\,, \\
		&C(0^2,6)\=\frac{5005}{15552}\;\approx\; 0.32182 \;>\;
		0.27006\;\approx\;\frac{175}{648}\=C(2^3)\,.
	\end{align}
\end{remark}

\begin{remark}
A main conjecture for BGW numbers in~\cite{GNYZ} was the monotonicity of these numbers for fixed $g$ in a suitable lexicographic ordering. However, this does not hold for the normalized numbers $C(\dd)$, e.g.,~$C(4,4)<C(3,5)$. 
\end{remark}

\begin{remark}
There is also a weak monotonicity conjecture~\cite{GNYZ} for BGW numbers, 
saying that the interval $I_{g,n}$ 
(the convex hull of normalized BGW numbers
of genus~$g$ and length~$n$) lies strictly to the left of~$I_{g,n+1}$. 
Even this weak monotonicity fails for normalized primitive Witten's intersection numbers $C(\dd)$, 
a counterexample being
\begin{align}
	C(2^5,8)=\frac{727759375}{2448880128}\approx 0.29718< 0.29746\approx \frac{419588015525}{1410554953728}=C(3^4,5).
\end{align}
\end{remark}

Following~\cite{GNYZ}, we also observe from Table~\ref{tablenormalizednumbers} and further numerical data that the normalized primitive Witten's intersection numbers $C(\dd)$ for each fixed~$g$ are close to each other, e.g. the minimum and maximum
values $C(37)$ and $C(2^{36})$ for $g=13$ differ by less than $2$~percent. In view of the nesting property, let us focus on the two values $C(3g-2)$ and $C(2^{3g-3})$. 
On one hand, the value of $C(3g-2)$ is given by~\eqref{onepoint},
which by Stirling's formula satisfies
\begin{align}\label{asymsmallest}
	C(3g-2)\;\sim\; \frac{1}{\pi}\Bigl(1-\frac{17}{36g}+\frac{1}{2592g^2}-\frac{557}{279936g^3}+\cdots\Bigr)\,, \quad g\to\infty \,.
\end{align}
On the other hand,
with the help of a deep result \cite{JoshiK, Kapaev} on a certain formal solution to the Painlev\'e I 
equation (cf. Section~\ref{secproofofthm1} for detail), one can obtain that
\begin{align}\label{asymlargest}
	C(2^{3g-3})\;\sim\; \frac{1}{\pi}\,\Bigl(1-\frac{2}{9\,g}-\frac{238}{2025\,g^2}-\frac{198149}{2733750\,g^3}+\cdots\Bigr) \,,\quad g\to\infty\,.
\end{align}
Conjecture~\ref{conjnesting} together with formulas~\eqref{asymsmallest},
\eqref{asymlargest} and \eqref{dilaton} implies that for $\dd\in(\mathbb{Z}_{\ge1})^n$ with $g(\dd)\in\mathbb{Z}_{\ge1}$, 
\beq\label{C(d)-1/pi}
\frac 1\pi \,-\, \frac{17}{36\pi\thin g(\dd)}\+ O\Bigl(\frac1{g(\dd)^2} \Bigr) \;\; \le \;\; C(\dd) 
\;\; \le  \;\; \frac 1\pi  \,-\, \frac2{9\pi\thin g(\dd)} \+ O\Bigl(\frac1{g(\dd)^2} \Bigr) 
\eeq
as $g(\dd)\to \infty$. 
This motivates the following 
\begin{thm}\label{thm1}
For $n\ge1$ and $\dd\in  (\mathbb{Z}_{\ge1})^n$ satisfying $g(\dd)\in\mathbb{Z}_{\ge1}$, we have
\beq
\label{STRONG} C(\dd) \= \dfrac1\pi \+ \text O\thin\Bigl(\dfrac1 {g(\dd)}\Bigr)
\eeq
uniformly as $g(\dd)\to\infty$. 
\end{thm}
\noindent Theorem~\ref{thm1} 
says that there exists an absolute constant~$K_1$ such that
\beq \Bigl|C(\dd) - \frac{1}{\pi} \Bigr| \,\leq \,  \frac{K_1}{g(\dd)} \label{Cdbound}
\eeq
holds for all $n\ge1$ and $\dd\in (\mathbb{Z}_{\ge1})^n$ with sufficiently large $g(\dd)\in\mathbb{Z}$. 

Our proof, which does not use~\eqref{asymlargest}, is along the lines of~\cite{Agg, GNYZ} and will be given in Section~\ref{secproofofthm1}.
It mainly uses the DVV relation~\eqref{DVV}, which written in terms of~$C(\dd)$ reads
\begin{align}\label{CVirasoro}
&C(\dd)\=\sum_{j=2}^n \frac{2d_j+1}{3\,(X(\dd)-1)}
\, C(d_2,\dots,d_j+d_1-1,\dots,d_n)\nn\\ 
&+\sum_{\substack{a, \thin b\\a+b=d_1-2}} \biggl[\frac{2}{3\,(X(\dd)-1)}\, C(a,b,d_2,\dots,d_n) \nn\\
&+\sum_{I\sqcup J=\{2,\dots,n\}} \frac{(X(a,\dd_{I})-1)! \, (X(b,\dd_{J})-1)!}{6\,(X(\dd)-1)!}\,
C(a,\dd_I) \,
C(b,\dd_J) \biggr]\,,
\end{align}
where 
\begin{align}\label{defX}
	X(\dd)\:\frac{1}{3}\sum_{j=1}^{n}(2d_j+1) \= 2 \, g(\dd)-2+n \,.
\end{align}
In particular, 
\begin{align}
	&C(0,\dd) \= \sum_{j=1}^n \frac{2d_j+1}{3(X(0,\dd)-1)}\,
	C(d_1,\dots,d_j-1,\dots,d_n)\,,\label{string}\\
	&C(1,\dd) \= C(\dd)\,. \label{dilaton}
\end{align}

An application of Theorem~\ref{thm1} will be given in Section~\ref{secproofofthm1}. 

Based on~\eqref{string} and Theorem~\ref{thm1} we will prove in Section~\ref{secproofofthm1} the following theorem.
\begin{thm}\label{thm2}
For $n\ge1$ and $\dd\in  (\mathbb{Z}_{\ge0})^n$ satisfying $g(\dd)\in\mathbb{Z}_{\ge1}$, we have
\beq\label{Cdasywithp0}
C(\dd) \= \dfrac1\pi \prod_{j=1}^{p_0(\dd)} \Bigl(1+\frac{2+j-p_0(\dd)}{3X(\dd)-3p_1(\dd)-3j}\Bigr) \+ O\Bigl(\thin\dfrac1 {\sqrt{g(\dd)}}\Bigr)\,,
\eeq
uniformly as $g(\dd)\to\infty$, where $p_i(\dd)$, $i\ge0$, denotes the multiplicity of~$i$ in~$\dd$.
\end{thm}
\noindent  
Theorem~\ref{thm2} says that there exists an absolute constant $K_2>0$ such that
\begin{align}\label{Cdasywithp0v2}
\biggl|C(\dd) \m \dfrac1\pi \prod_{j=1}^{p_0(\dd)} 
\Bigl(1+\frac{2+j-p_0(\dd)}{3X(\dd)-3p_1(\dd)-3j}\Bigr) \biggr| \;\leq\; 
\frac{K_2}{\sqrt{g(\dd)}}\,, 
\end{align}
for all $n\ge1$ and $\dd\in (\mathbb{Z}_{\ge0})^n$ satisfying $g(\dd)\in\mathbb{Z}_{\ge1}$.
Alternatively, we can write~\eqref{Cdasywithp0} as
\beq
C(\dd) \= \dfrac1\pi \, \frac{\bigl(\frac{2}{3}\bigr)^{p_0(\dd)}\,\bigl(\frac{3X-3p_0(\dd)-3p_1(\dd)+2}2\bigr)_{p_0(\dd)}}{(X-p_1(\dd)-p_0(\dd))_{p_0(\dd)}} \+  O\thin\Bigl(\dfrac1 {\sqrt{g(\dd)}}\Bigr)\,,
\eeq
as $g(\dd)\to\infty$. Here, $(a)_b:=a(a+1)\cdots(a+b-1)$ denotes the Pochhammer symbol.

As a direct application of Theorem~\ref{thm2}, we will prove the following
\begin{cor}\label{cor1}
As $g\to\infty$, for $k=O(\sqrt{g})$ we have the uniform leading asymptotics
\begin{align}\label{C02}
C(0^k,2^{3g-3+k}) \;\sim\; \frac{1}{\pi}\, e^{-\frac{k^2}{30g}}\,,
\end{align}
and for $k/\sqrt{g}\to \infty$ we have $C(0^k,2^{3g-3+k})\to0$.
\end{cor}
In terms of the original Witten's intersection numbers, \eqref{C02} reads
\begin{align}
\frac{2^{2g}\, 5^{3g-3+2k}}{3^{2g-2}\, (5g-6+2k)!}\,\int_{\overline{\mathcal{M}}_{g,\thin 3g-3+2k}} \psi_1^2\cdots \psi_{3g-3+k}^2 \; \sim \;
\frac{1}{\pi}\, e^{-\frac{k^2}{30g}}\,.
\end{align}
\begin{remark}
Theorem~\ref{thm2} implies Aggarwal's result~\eqref{asyAgg}.
Actually, Corollary~\ref{cor1} cannot be deduced from~\eqref{asyAgg}.
\end{remark}

Let us proceed to consider the full asymptotics of the normalized Witten's intersection numbers.
It was shown by Liu and Xu~\cite{LX} (cf.~also~\cite{GY}) that 
for fixed $n\ge1$ and fixed $\dd'=(d_1,\dots,d_{n-1})\in(\mathbb{Z}_{\ge0})^{n-1}$, 
the quotient $C(\dd',3g-3+n-|\dd'|)/C(3g-2)$ is a rational function of~$g$, 
and therefore has a full asymptotic expansion which is a power series of~$g^{-1}$. 
Then by using Stirling's formula and the 1-point formula~\eqref{onepoint} we know that 
$C(\dd',3g-3+n-|\dd'|)$ also has a full asymptotic expansion which is a power series of~$g^{-1}$. 
The polynomiality conjecture from~\cite{GY} predicts that for each~$k$ the coefficient of $g^{-k}$ in the large genus expansion of  
$C(\dd',3g-3+n-|\dd'|)$ is a polynomial of~$n$ and the multiplicities
of the arguments, and also that only the multiplicities of $0, 1, \dots, [3k/2]-1$ are involved. This conjecture  
was proved by Eynard {\it et al}~\cite{EGGGL} based on a formula proved in~\cite{Zhou, DYZ} (cf.~\cite{BDY}).

Following~\cite{GNYZ}, we will present and prove an improved version (see Theorem~\ref{thmpoly} below) of 
the polynomiality conjecture. 
For $n\ge1$, $\dd=(d_1,\dots,d_n)\in(\mathbb{Z}_{\ge0})^n$ satisfying $g(\dd)\in\mathbb{Z}$, we 
introduce a new normalization $\widehat{C}(\dd)$ of Witten's intersection numbers by
\begin{align}
\label{defChat}
\widehat{C}(\dd)\:\frac{C(\dd)}{\gamma(X(\dd))}\,,
	\end{align}
where
	\begin{align} \g(X) \: \frac{2^X}{\sqrt{\pi}\,3^{\frac{3X+1}2}}\,\frac{\Gamma\bigl(\frac{3X}2+1\bigr)}{\Gamma\bigl(\frac{X+3}2\bigr)\, \Gamma(X)}\,.
\end{align}
According to~\eqref{onepoint}, we know that $\widehat{C}(3g-2)\equiv1$ for any $g\ge1$. 
\begin{thm}\label{thmpoly} 
For any fixed $n$ and fixed $\dd'\in(\mathbb{Z}_{\ge2})^{n-1}$, the numbers $\widehat{C}(\dd)$ satisfy
\begin{align}\label{Chatdexpansionnew}
\widehat{C}(\dd) \;\sim\; \sum_k \frac{\widehat{c}_k(p_2(\dd'),p_3(\dd'),\dots)}{X(\dd)^k}\,, \qquad d_n\to\infty\,,
\end{align}
where $\dd=(\dd',d_n)$, $\widehat{c}_k$ are universal polynomials of $p_2,p_3,\dots$ having rational coefficients, 
with $\widehat{c}_0\equiv1$ and $\widehat{c}_k|_{p_b\equiv 0}=0$ $(k\ge1)$.
Moreover, under the degree assignments 
\begin{equation}\label{degreeassign}
\deg \, p_d \= 2d+1\quad (d\geq 1)\,,
\end{equation}
the polynomials $\widehat{c}_k$, $k\ge1$, satisfy the degree estimates
\begin{equation}\label{degreeest}
		\deg \, \widehat{c}_k \,\leq\, 3k-1\,.
	\end{equation} 
\end{thm}

\smallskip

\noindent {\bf Organization of the paper} 
In Section~\ref{secgeneral} we give a closed formula for Witten's intersection numbers. 
In Section~\ref{secproofofthm1} we prove Theorems \ref{thm1} and~\ref{thm2}.
In Section~\ref{secproofofthm2} we prove Theorem~\ref{thmpoly}. 

\smallskip

\noindent {\bf Acknowledgements} 
The authors would like to thank Felix Coy for pointing out an error in Theorem~\ref{thm2}, a typo in~\eqref{hatCk2}, and typos in Table~\ref{tablepoly}
in the first arXiv preprint version of the paper, and for his comments that were helpful for the proof of the 
corrected Theorem~\ref{thm2}. The work is supported by NSFC 12371254 and CAS YSBR-032.
Part of the work was done during visits of J.G. and D.Y. to MPIM, Bonn and visits of D.Z. to USTC; the authors thank both institutions for excellent working conditions.

\section{An explicit formula for Witten's intersection numbers}
\label{secgeneral}
In this section, 
based on~\eqref{fnmr} we derive an explicit formula 
for Witten's intersection numbers. 

By taking the Laurent expansion on the right-hand side of~\eqref{fnmr} with $n=2$, the following formula was obtained in~\cite{BDY}:
\begin{align}\label{twopoint}
\int_{\Mgtwobar} \psi_1^{d_1}\psi_2^{d_2} \= \frac{\sum_{l=0}^{d_1} \, (d_1+1-l)\,  
\xi_{l-1,3g-l}}{(2d_1+1)!!(2d_2+1)!!}\,,
\end{align}
where $g,d_1,d_2\ge0$ satisfying $d_1+d_2=3g-1$ and 
\begin{align*}
\xi_{k_1,k_2} \= \left\{
\begin{array}{ll}
\frac{(6g_1-5)!!\,(6g_2-5)!!}{2\cdot 24^{g_1+g_2-2}\,(g_1-1)!\,(g_2-1)!}\,,\quad & k_1=3g_1-2, k_2=3g_2-2\,,\\
-\frac{(6g_1-1)!!\,(6g_2-1)!!}{ 24^{g_1+g_2}\,g_1!\,g_2!}\,\frac{6g_2+1}{6g_2-1}\,,\quad &k_1=3g_1, k_2=3g_2-1\,,\\
\frac{(6g_1-1)!!\,(6g_2-1)!!}{2\cdot 24^{g_1+g_2}\,g_1!\,g_2!}\,\frac{6g_1+1}{6g_1-1}\,,\quad &k_1=3g_1-1, k_2=3g_2\,,\\
0\,, \quad &{\rm otherwise\,.}
\end{array}\right.
\end{align*}
Zograf~\cite{Zograf} derived another formula for two-point Witten's intersection numbers:
\begin{align}\label{twopointZograf}
	C(d_1,3g-1-d_1)\=\frac{1}{54^g\, (2g-1)! \, g}\,\sum_{d=-1}^{d_1-1} \eta_{g, \thin d}\,,
\end{align}
where $\eta_{g,d}$, $d\ge-1$, are defined by
\begin{equation}
\eta_{g,\thin d}\=(6g-3-2d)!!\, (2d+1)!!\cdot\left\{
	\begin{array}{ll}
		\frac{g-2j}{j! \, (g-j)!}\,, &\quad d=3j-1\,, \vspace{0.5ex}\\	-\frac{2}{j!\, (g-1-j)!}\,, &\quad d=3j\,, \vspace{0.5ex}\\	\frac{2}{j!\, (g-1-j)!}\,, &\quad d=3j+1\,.
	\end{array}
	\right.
\end{equation}
The equivalence of \eqref{twopoint} and \eqref{twopointZograf} was shown in~\cite{Guo}.

By taking the Laurent expansion on the right-hand side of~\eqref{fnmr} with $n\ge2$, a formula for Witten's intersection numbers was derived in~\cite{GY}, whose
expression written in terms of $C(\dd)$ is given in
the following proposition.
Before stating it, we introduce a notation: 
\begin{align}\label{defkappa}
a_{k_1,\dots,k_n}\:
\frac{2^{2g(\mathbf{k})}\, \,\mathrm{tr} \, A_{k_1}\cdots A_{k_n}}{3^{2g(\mathbf{k})+n-2}\, (2g(\mathbf{k})+n-3)!}\,,
\end{align}
for ${\bf k}=(k_1,\dots,k_n) \in \ZZ_{\ge-1}^n$, 
and $a_{k_1,\dots,k_n}$ are defined as~0 if some of $k_j$ is less than or equal to~$-2$.
Here
$$ M(\lambda)\;=:\;\sum_{k\ge-1} A_k \, \lambda^{-k}\,.$$

\begin{prop}[\cite{GY}]\label{propformulaWK1}
For $n\ge 2$ and $\dd=(d_1,\dots,d_n)\in (\ZZ_{\ge 0})^n$, we have
\begin{align}\label{formulaWKeq1}
C({\bf d}) \= \sum_{\substack{\sigma\in S_{n} \\ \sigma(n)=n}} (-1)^{|S_{\sigma}^-|+1}
\sum_{\substack{\underline{J}\in (\mathbb{Z}+\frac12)^n \\ \{1\leq q\leq n\mid J_q>0\}=S_{\sigma}^{-}} }
a_{d_{\sigma(1)}+J_1-J_n,\dots,d_{\sigma(n)}+J_n-J_{n-1}} \,.
\end{align}
\end{prop}
The following proposition is an analogue of~\cite[Proposition 3]{GNYZ}.
\begin{prop}\label{propformulaWK2}
For $n\ge 2$ and $\dd=(d_1,\dots,d_n)\in (\mathbb{Z}_{\ge 0})^n$,
\begin{align}\label{formulaWKeq2}
C(\dd) \= \sum_{\substack{\sigma\in S_{n} \\ \sigma(n)=n}} (-1)^{|S_{\sigma}^{-}|+1} \sum_{\substack{k_1,\dots,k_n\ge-1\\ k_1+\cdots+k_n=d_1+\cdots+d_n}} 
a_{k_1,\dots,k_n}\,\omega_{\dd,\sigma,\mathbf{k}}\,,
\end{align}
where the sets $S_{\sigma}^{+},\,S_{\sigma}^{-}\subseteq \{1,\dots,n\}$ are denoted as
\begin{align*}
S_{\sigma}^+\=\bigl\{1\leq r\leq n \,\big|\, \sigma(r+1)>\sigma(r)\bigr\}\,,\quad 	S_{\sigma}^-\=\bigl\{1\leq r\leq n \,\big|\, \sigma(r+1)<\sigma(r)\bigr\}\,,
\end{align*} and the numbers $\omega_{\dd,\sigma,\mathbf{k}}$ have the following explicit expression
\begin{align}
\omega_{\dd,\sigma,\mathbf{k}}\=\max\biggl\{0\,,\,\min_{r\in S_{\sigma}^+}\biggl\{\sum_{q=1}^r(d_{\sigma(q)}-k_q)\biggr\} 
+\min_{r\in S_{\sigma}^-}
\biggl\{\sum_{q=1}^r(k_q-d_{\sigma(q)})\biggr\} \biggr\}\,.
\end{align}
\end{prop}
\noindent The proof is almost identical to that of~\cite[Proposition~3]{GNYZ} and is
therefore omitted.

Let us give some examples for Proposition~\ref{propformulaWK2}. The $n=2$ case of Proposition~\ref{propformulaWK2} is equivalent to~\eqref{twopoint}.
For $n\ge1$, $\mathbf{e}\in\mathbb{Z}^n$, introduce the notation
\begin{align}\label{defM}
M(\mathbf{e}) \:\max\bigl\{0,\min_{1\leq i\leq n}\{e_i\}\bigr\}\,,
\end{align}
which has the generating function
	\begin{align}\label{generatingM}
		\sum_{e_1,\dots,e_n\ge0}x_1^{e_1-1}\cdots x_n^{e_n-1}\,M(e_1,\dots,e_n)\=\frac{1}{(1-x_1\cdots x_n)\prod_{i=1}^n(1-x_i)}\,.
	\end{align}
Using~\eqref{defM}, the $n=3$ case of Proposition~\ref{propformulaWK2} can be written as
\begin{align}\label{threepointWK}
	C(d_1,d_2,d_3)\=2\sum_{k_1+k_2+k_3=d_1+d_2+d_3 } a_{k_1,k_2,k_3} \, M(d_1-k_1,d_1+d_2-k_1-k_2)\,,
\end{align}
and the $n=4$ case of Proposition~\ref{propformulaWK2} can be written as
\begin{align}
	&C(d_1,d_2,d_3,d_4)\=2\sum_{k_1+k_2+k_3+k_4=d_1+d_2+d_3+d_4}a_{k_1,k_2,k_3,k_4} \nn\\
	&\quad \times \Bigl(M\bigl(d_1-k_1,d_1+d_2-k_1-k_2,k_4-d_4\bigr) \nn\\
	&\qquad -M\bigl(d_1-k_2,d_1+d_2-k_2-k_3,d_1+d_3-k_1-k_2,k_4-d_4\bigr) \nn\\
	&\qquad -M\bigl(d_1-k_1,d_2-k_3,k_2-d_3,k_4-d_4\bigr)\Bigr)\,, \label{fourpointWK}
\end{align}
where we have used the fact that $a_{k_1,k_2,k_3,k_4}=a_{k_2,k_3,k_4,k_1}=a_{k_1,k_4,k_3,k_2}$.

\section{Uniform large genus asymptotics of Witten's intersection numbers} \label{secproofofthm1}
In this section, we prove Theorem~\ref{thm1}. Our proof will mainly use the DVV relation~\eqref{DVV}, 
the techniques introduced by Aggarwal~\cite{Agg} (see also~\cite{GNYZ}), and the improvement in~\cite{GNYZ}.

\subsection{Lower Bound}

In this subsection we give lower bounds for the normalized Witten's intersection numbers $C(\dd)$.
\begin{lemma}\label{lemmalowerbound0}
For $n\ge1$ and for $\dd=(d_1,\dots,d_n)\in(\mathbb{Z}_{\ge0})^n$ satisfying $g(\dd)\in\ZZ_{\ge0}$ and $2g(\dd)-2+n>0$, we have
\begin{align}\label{lowerbound0}
	C(\dd)\;>\; 0\,.
\end{align}
\end{lemma}
\begin{proof}
By using~\eqref{CVirasoro} and by recalling that $C(0^3)=1/3$, $C(1)=1/6$.
\end{proof}

The following lemma proves one side of Conjecture~\ref{conjnesting}.

\begin{lemma}\label{lemmalowerbound}
For $g,n\ge1$ and $\dd\in(\mathbb{Z}_{\ge1})^n$ satisfying $|\dd|=3g-3+n$, we have
\begin{align}\label{lowerbound}
C(\dd)\;\geq\; C(3g-2)\,.
\end{align}
\end{lemma}
\begin{proof}
For $n=1$, inequality~\eqref{lowerbound} is trivial. For $n=2$ inequality~\eqref{lowerbound} was obtained in~\cite{DGZZ20,DGZZ21}. 
Now assume $n\ge3$, and let us prove~\eqref{lowerbound} by induction on $X(\dd)$. For $X(\dd)\leq 5$, inequality~\eqref{lowerbound} can be checked directly from Table~\ref{tablenormalizednumbers}. For $\dd=(d_1,\dots,d_n)\in(\mathbb{Z}_{\ge1})^n$ with $X(\dd)\ge6$, because of the symmetry of~$C$-values and formula~\eqref{dilaton}, we can assume that $2\leq d_1\leq\dots\leq d_n$. For $d_1\ge3$, by using \eqref{onepoint}, \eqref{CVirasoro} and the inductive assumption we have
\begin{align}\label{Cdlowerboundproof}
&C(\dd)\;\ge\; \sum_{j=2}^n\frac{2d_j+1}{3(X(\dd)-1)}\,C(3g-2)+\frac{2(d_1-1+\frac{4}{3\,(X(\dd)-2)})}{3\,(X(\dd)-1)}\,   C(3g-5) \nn\\
&\=C(3g-2) \m\frac{2(d_1-1)}{3(X(\dd)-1)}\Bigl(C(3g-2)-C(3g-5)\Bigr)
+\frac{8\,C(3g-5)}{9\,(X(\dd)-1)\,(X(\dd)-2)}  \nn\\
&\=C(3g-2) \+ \frac{2\, C(3g-5)}{3(X(\dd)-1)} \, \Bigl(-\frac{(d_1-1)(34g-35)}{36g(2g-3)(g-1)}+\frac{4}{3(X(\dd)-2)}\Bigr)\,,
\end{align}
where in the inequality we have omitted the quadratic-in-$C$ term since they are nonnegative and have used the following estimate which is implied by equation~\eqref{string}
\begin{align}\label{C0dlowerbound}
C(0,d_1-2,d_2,\dots,d_n)\;\ge\;\Bigl(1+\frac{2}{3(X(\dd)-2)}\Bigr)\,  C(3g-5)\,.
\end{align}
The expression in the parenthesis of the right-hand side of~\eqref{Cdlowerboundproof} is nonnegative since
\begin{align}\label{est1}
-\frac{(d_1-1)(34g-35)}{36g(2g-3)(g-1)}+\frac{4}{3(X(\dd)-2)} \;\ge\; -\frac{(34g-35)}{12ng(2g-3)}+\frac{4}{3(2g-4+n)}
\nn\\ \;\ge\;\frac{1}{3n}\Bigl(-\frac{(34g-35)}{4g(2g-3)}+\frac{12}{(2g-4+3)}\Bigr) \;\ge\;0\,,
\end{align}
where we have used the conditions $d_1-1\leq \frac{3g-3}n$, $3\leq n\leq 3g-3$ and $g\ge3$ (since $X(\dd)=2g-2+n\ge6$). Inequalities~\eqref{Cdlowerboundproof} and \eqref{est1} imply that $C(\dd)\ge C(3g-2)$ when $d_1\ge3$.

For $d_1=2$, we have the estimate
\begin{align}\label{C00dlowerbound}
&C(0,0,d_2,\dots,d_n)\=\sum_{j=2}^n\frac{2d_j+1}{3(X(\dd)-2)} C(0,d_2,\dots,d_j-1,\dots,d_n) \\
&\qquad \ge\; \Bigl(1+\frac{1}{3(X(\dd)-2)}\Bigr) \, \Bigl(1+\frac{2}{3(X(\dd)-3)}\Bigr) \, C(3g-5) \;\ge\;  C(3g-2)\,.\nn
\end{align}
where for the equality we have used~\eqref{string}, for the first inequality we have used the inductive assumption and an analogue of~\eqref{C0dlowerbound}, and for the second inequality we have used the one point forumla~\eqref{onepoint} and the condition $n\le 3g-3$ and $g\ge3$. By using~\eqref{CVirasoro}, \eqref{C00dlowerbound}, we obtain that
\begin{align}
&C(d_1=2,\dots,d_n) \;\ge\; \sum_{j=2}^n\frac{2d_j+1}{3(X(\dd)-1)} C(d_2,\dots,d_j+1,\dots,d_n) \nn\\
&\qquad\qquad\qquad\qquad  \+ \frac{2}{3(X(\dd)-1)}\, C(0,0,d_2,\dots,d_n) \;\ge\; C(3g-2)\,.
\end{align} 
This completes the proof.
\end{proof}
\begin{cor}
For $g,n\ge1$ and $\dd\in(\mathbb{Z}_{\ge0})^n$ satisfying $|\dd|=3g-3+n$, we have
\begin{align}\label{lowerbound1}
	C(\dd)\;\geq\; C(3g-2)\, \prod_{j=1}^{p_0(\dd)}\Bigl(1+\frac{j+2-p_0(\dd)}{6g-6+3n-3j-3p_1(\dd)}\Bigr)\,,
\end{align}
where $p_0(\dd)$ denotes the multiplicity of~$0$ in~$\dd$.
\end{cor}
\begin{proof}
By~\eqref{string} and Lemma~\ref{lemmalowerbound} (some details can be found in
the proof of Theorem~\ref{thm2} below).
\end{proof}

We mention that several other lower bounds for Witten's intersection numbers are obtained in~\cite[Proposition~4.1]{Agg}, \cite[Theorem~5]{DGZZ20} and \cite[Corollary~5.4]{LX07}.

\subsection{Upper Bound}
Analogous to~\cite[Lemma~3.1]{Agg}, we first prove the following lemma. 
\begin{lemma}\label{estimatethird}
For $n\ge1$, $g\ge2$ and for $\mathbf{d}\in(\mathbb{Z}_{\ge2})^n$ satisfying $|\dd|=3g-3+n$, we have
\begin{align}\label{upperboundthirdterm}
\sum_{\substack{a,b\ge0\\a+b=d_1-2}}\sum_{\substack{I\sqcup J=\{2,\dots,n\}\\ X(a,\dd_I),X(b,\dd_J)\in\ZZ}}\frac{\bigl(X(a,\mathbf{d}_{I})-1\bigr)!\,\bigl(X(b,\mathbf{d}_{J})-1\bigr)!}{6\,\bigl(X(\mathbf{d})-1\bigr)!}\;\leq\; \frac{2}{(X(\mathbf{d})-1)(X(\mathbf{d})-2)}\,.
\end{align}
\end{lemma}
\begin{proof}
For $n=1$, the inequality~\eqref{upperboundthirdterm} is trivial. Assume that $n\ge2$. For each $a,b,I,J$ satisfying $a+b=d_1-2$, $I\sqcup J=\{2,\dots,n\}$ and $X(a,\dd_I),X(b,\dd_J)\in\ZZ$, we set $n_1=|I|$, $g_1=g(a,\dd_I)=(a+|\dd_I|+2-n_1)/3\in\ZZ$, $n_2=|J|$, $g_2=g(b,\dd_J)=(b+|\dd_J|+2-n_2)/3\in\ZZ$. By counting the number of 4-tuples $(a,b,I,J)$ with given values of $n_i$ and $g_i$, we find
\begin{align}
&\sum_{\substack{a,b\ge0\\a+b=d_1-2}}\sum_{\substack{I\sqcup J=\{2,\dots,n\}\\ X(a,\dd_I),X(b,\dd_J)\in\ZZ}}\frac{\bigl(X(a,\mathbf{d}_{I})-1\bigr)!\,\bigl(X(b,\mathbf{d}_{J})-1\bigr)!}{6\,\bigl(X(\mathbf{d})-1\bigr)!} \nn\\
\;\leq\; &\sum_{n_1+n_2=n-1}\sum_{\substack{g_1,g_2\ge1\\g_1+g_2=g }}\binom{n-1}{n_1}\frac{(2g_1+n_1-2)!\,(2g_2+n_2-2)!}{6\,(X(\mathbf{d})-1)!}\nn\\
\= &\sum_{n_1+n_2=n-1}\binom{n-1}{n_1}\biggl(\frac{n_1!\,(X(\dd)-n_1-3)!}{3\,(X(\mathbf{d})-1)!} \,+\,  \sum_{\substack{g_1,g_2\ge2\\g_1+g_2=g }}\frac{(2g_1+n_1-2)!\,(2g_2+n_2-2)!}{6\,(X(\mathbf{d})-1)!}\biggr)\,.
\label{estimatethirdterm}
\end{align}
We estimate the two terms on the right-hand side of~\eqref{estimatethirdterm} separately. For the first term we have 
\begin{align}\label{estimatethirdterm2}
&\sum_{n_1+n_2=n-1}\binom{n-1}{n_1}\frac{n_1!\,(X(\mathbf{d})-n_1-3)!}{3\,(X(\mathbf{d})-1)!}
\=\frac{1}{3\,(X(\mathbf{d})-1)\,(X(\mathbf{d})-2)}\sum_{n_1=0}^{n-1}\prod_{j=1}^{n_1}\frac{n-j}{X(\mathbf{d})-2-j}\nn \\
&\qquad\qquad \;\leq\; \frac{1}{3\,(X(\mathbf{d})-1)\,(X(\mathbf{d})-2)}\biggl(2+\sum_{n_1=2}^{\infty}\Bigl(\frac{3}{5}\Bigr)^{n_1-2}\biggr)\=\frac{3}{2\,(X(\mathbf{d})-1)\,(X(\mathbf{d})-2)}\,, 
\end{align}
where in the inequality we have used the fact that $n\leq \frac{3X(\mathbf{d})}{5}$ (implied by  $\dd\in(\mathbb{Z}_{\ge2})^n$) and $n\leq X(\dd)-2$ (implied by $g\ge2$).
For the second term, we note that it vanishes if $g\le 3$, otherwise it has the upper bound
\begin{align}
&\sum_{n_1+n_2=n-1}\sum_{\substack{g_1,g_2\ge2\\g_1+g_2=g }}\binom{n-1}{n_1}\frac{(2g_1+n_1-2)!\,(2g_2+n_2-2)!}{6\,(X(\mathbf{d})-1)!}\nn\\
\leq\; &\sum_{\substack{g_1,g_2\ge2\\g_1+g_2=g \\n_1+n_2=n-1 }} \frac{\binom{2g-4}{2g_1-2}^{-1}}{6\,(X(\mathbf{d})-1)(X(\mathbf{d})-2)} \;\leq\;  \frac{n\, (g-3) \,\binom{2g-4}{2}^{-1}}{6\,(X(\mathbf{d})-1)(X(\mathbf{d})-2)} \nn\\
\leq\; & \frac{3}{8\,(X(\mathbf{d})-1)\,(X(\mathbf{d})-2)}\,,\label{estimatethirdterm1}
\end{align}
where for the first inequality we used the fact that $X(\dd)-3=(2g_1+n_1-2)+(2g_2+n_2-2)$ and that
	\begin{align}
		\binom{a_1}{b_1}\binom{a_2}{b_2}\;\leq\; \binom{a_1+a_2}{b_1+b_2}\,,
	\end{align}
	and for the last inequality we used $g\ge4$ and $n\leq 3g-3$. 
	Combining \eqref{estimatethirdterm},~\eqref{estimatethirdterm2},~\eqref{estimatethirdterm1}, we obtain the lemma.
\end{proof}

Following again the idea of Aggarwal~\cite{Agg} (cf.~\cite{GNYZ}), we now define 
\begin{align}\label{defthetagn}
\theta_{\X, \thin n}\=
\max_{\substack{\dd\in(\ZZ_{\ge1})^n \\ X(\dd)=\X}}C(\dd)\,, \quad \X,n\ge1\,.
\end{align}
According to~\eqref{dilaton} we have 
\begin{align}\label{thetaineq0}
\theta_{\X, \thin n} \geq\theta_{\X-1, \thin n-1}\,, \qquad \forall\, \X,n\ge2\,.
\end{align}
We have the following 
\begin{lemma}\label{lemmaC0C00}
For $n\ge1$, $\dd\in(\mathbb{Z}_{\ge1})^n$ we have
\begin{align}
&C(0,\dd)\;\leq\;
\frac{3X(0,\dd)-3p_1(\dd)-1}{3X(0,\dd)-3p_1(\dd)-3}\, \theta_{X(0,\dd)-p_1(\dd)-1,n-p_1(\dd)}\,,\label{C0dest}\\ &C(0,0,\dd)\;\leq\;\frac{(3X(0^2,\dd)-3p_1(\dd)-2)(3X(0^2,\dd)-3p_1(\dd)-7)}
{(3X(0^2,\dd)-3p_1(\dd)-3)(3X(0^2,\dd)-3p_1(\dd)-9)}\, \theta_{X(0^2,\dd)-p_1(\dd)-2,n-p_1(\dd)}\,. \label{C00dest}
\end{align}
\end{lemma}
\begin{proof}
Let us first prove~\eqref{C0dest}. Since both sides of~\eqref{C0dest} are invariant if any argument~$1$ in $\dd$ is removed, it is sufficient to prove~\eqref{C0dest} for the case when $\dd\in(\ZZ_{\ge2})^n$. By using~\eqref{string}, \eqref{dilaton} and \eqref{thetaineq0} we have
\begin{align}
&C(0,\dd) \= \sum_{j=1}^{n} \frac{2d_j+1}{3X(0,\dd)-3} C(d_1,\dots,d_j-1,\dots,d_n)\;\leq\; \frac{3X(0,\dd)-1}{3X(0,\dd)-3} \, \theta_{X(0,\dd)-1,n}\,,
\end{align}
which proves inequality~\eqref{C0dest}.
To prove~\eqref{C00dest}, for the same reason we can assume $\dd\in(\ZZ_{\ge2})^n$. By using~\eqref{string}, \eqref{thetaineq0} and \eqref{C0dest} we have
\begin{align}
&C(0^2,\dd) \= \sum_{j=1}^{n} \frac{2d_j+1}{3X(0^2,\dd)-3} C(0,d_1,\dots,d_j-1,\dots,d_n)\nn\\
&\;\leq\; \frac{3X(0^2,\dd)-2}{3X(0^2,\dd)-3} \, \max\biggl\{\frac{3X(0^2,\dd)-4}{3X(0^2,\dd)-6}\,\theta_{X(0^2,\dd)-2,n},
\frac{3X(0^2,\dd)-7}{3X(0^2,\dd)-9}\,\theta_{X(0^2,\dd)-3,n-1}\biggr\}\nn\\
&\;\leq\; \frac{3X(0^2,\dd)-2}{3X(0^2,\dd)-3} \,  
\frac{3X(0^2,\dd)-7}{3X(0^2,\dd)-9}\,\theta_{X(0^2,\dd)-2,n}\,,
\end{align}
which finishes the proof of~\eqref{C00dest}.
\end{proof}

Following~\cite{Agg} (see also~\cite{GNYZ}), introduce a number-theoretic function $f(\X,n)$, defined through the recursion
\begin{align}\label{defg*}
f(\X,n)&\=
\frac{2}{3}\,f(\X-1,n-1)\+\frac{1}{3}\,f(\X-1,n+1) \+\frac{4}{(\X-1)(\X-2)}\,,
\end{align} 
for any $n\ge3$, $\X\ge8$ together with the initial data $f(\X,n)=1/\pi$
for $1\leq \X\leq 7$ or $n=1,2$. We notice that the function $f(\X,n)$ is the same as
the one introduced in~\cite{GNYZ}. The following proposition is proved in~\cite{GNYZ}.
\begin{lemma}[\cite{GNYZ}]\label{lemmagXn}
The function $f(\X,n)$ satisfies the following properties.
\begin{enumerate}
\item $1/\pi\leq f(\X,n)\leq 1$ for any $n\ge1$, $X\ge1$.
\item $f(\X,n)$ is monotone increasing with respect to~$n$.
\item For $1\leq n\leq \X/5$, as $\X\to\infty$, $f(\X,n)=1/\pi+O(1/\X)$ where the $O$-constant is uniform in~$n$.
\end{enumerate}
\end{lemma}

The significance of the function $f(\X,n)$ is given by the following important lemma.
\begin{lemma}\label{lemineqthetaXn}
For $n\ge1$, $\X\ge1$, the numbers $\theta_{\X,n}$ have the upper bound
\begin{align}\label{ineqthetaXn}
\theta_{\X,n}\;\leq\; f(\X,n)\,.
\end{align}
\end{lemma}
\begin{proof}
For $1\leq \X\leq 7$, we check individually that inequality~\eqref{ineqthetaXn} holds. For $n=1$, by~\eqref{onepoint} inequality~\eqref{ineqthetaXn} is equivalent to the inequality
\begin{align}\label{onepointbound}
\frac{3\, (6g-3)!!}{54^g\, g!\, (2g-2)!}\;\leq\; \frac{1}{\pi}\,, \quad \forall\, g\ge1\,,
\end{align}
which is true because the left-hand side is monotone increasing and asymptotic to~$\frac{1}{\pi}$ by Stirling's formula.
For $n=2$, according to~\cite{DGZZ20} we have for $g\ge1$, $d_1,d_2\ge0$ satisfying $d_1+d_2=3g-1$,
\begin{align}
C(d_1,d_2) \;\leq\; C(0,3g-1) \,.
\end{align}
Applying~\eqref{string} and~\eqref{onepoint}, we have
\begin{align}
C(d_1,d_2)\;\leq\;\frac{6g-1}{6g-3}\,\frac{3\, (6g-3)!!}{54^g\, g!\, (2g-2)!} \;\leq\; \frac{1}{\pi}\,,
\end{align}  where the second inequality holds for the similar reason as~\eqref{onepointbound}. 
	
Now consider the case that $n\ge3$ and $\X\ge8$. Let us prove inequality~\eqref{ineqthetaXn} by induction on~$\X$. For $\X=8$, it is easy to check the validity of inequality~\eqref{ineqthetaXn}.
Consider $\X\ge9$ and $\dd=(d_1,\dots,d_{n})\in \bigl(\mathbb{Z}_{\ge1}\bigr)^{n}$ satisfying $X(\dd)=\X$. Without loss of generality, we assume $d_1=\min \{d_j\}$. For $d_1\ge3$, 
by using \eqref{CVirasoro}, Lemma~\ref{estimatethird} and \eqref{C0dest} we can obtain
\begin{align*}
C(\dd)\;
&\leq\; \Bigl(1-\frac{2d_1-2}{3\X-3}\Bigr) \, f(\X-1,n-1)\+\frac{2d_1-6}{3\X-3} \, f(\X-1,n+1)\nn\\
&\qquad \+\frac{4}{3\X-3}\,\frac{3\X-4}{3\X-6}\, f(\X-2,n) \+\frac{2}{(\X-1)(\X-2)}\,.
\end{align*}
So 
\begin{align}
C(\dd)&\;\leq\; \Bigl(1-\frac{2d_1-2}{3\X-3}\Bigr) \, f(\X-1,n-1)\+\frac{2d_1-2}{3\X-3} \, f(\X-1,n+1)\+\frac{26}{9(\X-1)(\X-2)}\nn\\
&\leq\; \frac{2}{3} \, f(\X-1,n-1)\+\frac{1}{3} \, f(\X-1,n+1)\+\frac{26}{9(\X-1)(\X-2)}\,, \label{Cdestd1ge3}
\end{align}
where for the first inequality we used the facts that $f(\X-2,n)\leq f(\X-1,n+1)$ and $f(\X-1,n+1)\leq 1$, and for the second inequality we used the facts that $d_1\leq \frac{3\X-3}{2n}$, $n\ge3$ and $f(X-1,n+1)\ge f(\X-1,n-1)$.
For $d_1=2$, by using \eqref{CVirasoro}, Lemma~\ref{estimatethird} and \eqref{C00dest} we obtain
\begin{align}
C(\dd)\;
&\leq\; \Bigl(1-\frac{2}{3\X-3}\Bigr) \, f(\X-1,n-1)\+\frac{2}{(\X-1)(\X-2)}\nn\\ 
&\qquad \+\frac{2}{3X-3}\,\frac{3\X-5}{3\X-6}\,\,\frac{3\X-10}{3\X-12} f(\X-3,n-1)\,.\nn
\end{align}
So
\begin{align}
C(\dd)\;
\leq\; \frac{2}{3}\,f(\X-1,n-1)\+\frac{1}{3} \, f(\X-1,n+1)\+\frac{79}{27(\X-1)(\X-2)}\,,\label{Cdestd1=2}
\end{align}
where we used the facts that $f(\X-3,n-1)\leq f(\X-1,n+1)$, $f(\X-1,n+1)\leq 1$ and $f(X-1,n+1)\ge f(\X-1,n-1)$. For $d_1=1$, by using~\eqref{dilaton} we have
\begin{align}\label{Cdestd1=1}
C(\dd) \;\leq\; f(\X-1,n-1)\;\leq\; \frac{2}{3}\,f(\X-1,n-1)\+\frac{1}{3} \, f(\X-1,n+1)\,.
\end{align}
From~\eqref{Cdestd1ge3}, \eqref{Cdestd1=2}, \eqref{Cdestd1=1}, we get the validity of~\eqref{ineqthetaXn}.
The lemma is proved. 
\end{proof}

\begin{lemma}\label{lemmathetaineq}
For $\X\ge4$ and for $2\leq n\leq \frac{3\X}5$, we have
\begin{align}\label{thetaineq}
\theta_{\X-3,n-1} \;\leq\; \theta_{\X,n}\+\frac{M}{\X-1}
\end{align}
for some constant $M>0$.
\end{lemma}
\begin{proof}
To prove~\eqref{thetaineq}, it suffices to prove that for $n\ge1$, $\dd\in(\ZZ_{\ge2})^n$,
\begin{align}\label{C4ineq}
C(4,d_1,\dots,d_n) \;\ge\; C(d_1,\dots,d_n)\m\frac{M}{X(4,d_1,\dots,d_n)-1}
\end{align}
for some constant $M>0$. Indeed, if~\eqref{C4ineq} is true for $\dd\in(\ZZ_{\ge2})^n$ then it is also true for $\dd\in\ZZ_{\ge 1}^n$ by changing $M$ to some other constant $M'$ because $n\leq \frac{3\X}5$.
According to e.g.~\cite{BDY,Dickey, DYZ},  for a KdV tau-function $\tau=\tau(\mathbf{t})$ the following identities hold:
\begin{align}
& \e^2 \frac{\partial^2 \log \tau}{\partial t_{1}\partial t_{1}}
 \=  \frac{u_{4}}{144}  \+\frac{uu_{2}}{6} \+ \frac{u_1^2}{24} \+ \frac{u^3}{3} \,, \label{Omega11}\\
& \e^2 \frac{\partial^2 \log \tau}{\partial t_{1} \partial t_{4}}
 \= \frac{u_{10}}{2903040}+\frac{u\,u_{8}}{60480}+\frac{13 u^2\, u_{6}}{40320}
+\frac{7 u^3\, u_{4}}{2160}+\frac{17 u_{4}^2}{96768}+\frac{u\,u_{3}^2}{448} 
+\frac{5u^4 u_{2}}{288}\nn\\
& +\frac{41 u^2 u_{2}^2}{2880}+\frac{79 u_{2}^3}{40320}+\frac{u^3 u_1^2}{36} +\frac{u_1^4}{384} +\frac{u_{7} u_1}{17280}+\frac{41 u_{2}\,u_{6}}{241920}+\frac{23 u_{3} u_{5}}{80640}+\frac{17 u\,u_1\,u_{5}}{10080}\nn\\
&+\frac{13 u\,u_{2}\,u_{4}}{3360}+\frac{251 u_1^2\,u_{4} }{120960}+\frac{1}{60} u^2\,u_1\,u_{3}+\frac{1}{36} u\,u_1^2\,u_{2}+\frac{151 u_1\,u_{2}\,u_{3}}{20160}+\frac{u^6}{144}\,,
\label{Omega14}
\end{align}
where $u=\partial_x^2(\log\tau)$, $u_k=\p_{t_0}^k (u)$. 
Similar to~\cite{LX10}, using~\eqref{Omega11} one can obtain the following recursion for $C(\dd)$:
\begin{align}\label{Cdrecursion1}
&C(\dd) \= C(0^6,\dd) \+ \frac{3}{2} \, \sum_{I\sqcup J=\{1,\dots,n\}} 
\frac{(X(0^3,d_I)-1)!(X(0^3,d_J)-1)!}{(X(\dd)+1)!}\, C(0^3,d_I)\, C(0^3,d_J)\nn\\
&\+ 6 \, \sum_{I\sqcup J=\{1,\dots,n\}} 
\frac{(X(0^2,d_I)-1)!(X(0^4,d_J)-1)!}{(X(\dd)+1)!}\, C(0^2,d_I)\, C(0^4,d_J)\nn\\
&\+ 3 \, \sum_{I\sqcup J\sqcup K=\{1,\dots,n\}} 
\frac{(X(0^2,d_I)-1)!(X(0^2,d_J)-1)!(X(0^2,d_K)-1)!}{(X(\dd)+1)!}\, C(0^2,d_I)\, C(0^2,d_J)\, C(0^2,d_K)\,,
\end{align}
which can imply that for $n\ge1$, $\dd\in(\ZZ_{\ge2})^n$,
\begin{align}\label{C0^6d}
C(0^6,\dd) \;\ge\; C(\dd) \m \frac{M_1}{X(\dd)+1}\,, \qquad C(0^{12},\dd) \;\ge\; C(0^6,\dd) \m \frac{M_2}{X(\dd)+1}
\end{align}
for some constant $M_1,M_2>0$. Similarly, the identity~\eqref{Omega14} leads to the inequality 
\begin{align}\label{C4d}
C(4,\dd) \;\ge\; C(0^{12},\dd)\,.
\end{align} 
By applying~\eqref{C0^6d} and \eqref{C4d}, we obtain~\eqref{C4ineq}. The lemma is proved.
\end{proof}

We are ready to prove Theorem~\ref{thm1}.
\begin{proof}[Proof of Theorem~\ref{thm1}]
For $n\ge1$, $\dd\in(\mathbb{Z}_{\ge1})^n$ satisfying $g(\dd)\in\mathbb{Z}$, by Lemma~\ref{lemmalowerbound} and~\eqref{defthetagn} we have 
\begin{align}\label{lowerupperbound}
C(3g(\dd)-2)\;\leq\; C(\dd) \;\leq\; \theta_{X(\dd), \thin n}\,.
\end{align}
On one hand, we know from~\eqref{asymsmallest} that $C(3g(\dd)-2)=\frac{1}{\pi}+O(\frac{1}{g(\dd)})$ with an absolute $O$-constant. On the other hand, according to Lemma~\ref{lemineqthetaXn},  $\theta_{X(\dd),n}\leq f(X(\dd),n)$, then by Lemma~\ref{lemmagXn} we know that when $n\leq \frac{\X}{5}$, $\theta_{\X,n}=\frac{1}{\pi}+O(\frac{1}{\X})$ with an absolute $O$-constant. 

Consider $\dd\in(\mathbb{Z}_{\ge1})^n$ such that $\frac{X(\dd)}5<n< \frac{3X(\dd)}{5}$. Without loss of generality, assume $d_1=\min\{d_j\}$. We then conclude that $d_1\in\{1,\dots,6\}$. 
Using the recursion~\eqref{CVirasoro} and using an estimate similar to~\eqref{Cdestd1ge3}, \eqref{Cdestd1=2}, we have
\begin{align}
&C(\dd) \;\leq\; \Bigl(1-\frac{10}{3X(\dd)-3}\Bigr)\, \theta_{X(\dd)-1, \thin n-1} \+ \frac{10}{3X(\dd)-3}\,
\max\{\theta_{X(\dd)-1, \thin n-1},\theta_{X(\dd)-3,\thin n-1}\}\nn\\
&\qquad\qquad \+ \frac{M_2}{(X(\dd)-1)(X(\dd)-2)}
\end{align}
for some $M_2>0$. 
By applying Lemma~\ref{lemmathetaineq}, we obtain
\begin{align}\label{thetaineq2}
\theta_{\X, \thin n}\;\leq\; \theta_{\X-1, \thin n-1} \+ \frac{M_3}{(\X-1)(\X-2)}
\end{align}
for some constant $M_3>0$.
Iterating~\eqref{thetaineq2} $t=[\frac{5n-\X+3}4]$ times and 
using Lemma~\ref{lemmagXn} and Lemma~\ref{lemmathetaineq} we obtain that $\theta_{\X,n}$ is bounded by $\frac{1}{\pi}+O(\frac{1}{\X})$ uniformly when $\frac{\X}{5}< n\leq \frac{3\X}{5}$. This together with~\eqref{lowerupperbound} implies that 
formula~\eqref{STRONG} holds for $\frac{X(\dd)}{5}\leq  n\leq \frac{3X(\dd)}{5}$. 
For $n>\frac{3X(\dd)}{5}$, we have $\theta_{X(\dd), \thin n}=\theta_{X(\dd)-1,\thin n-1}$ 
(indeed, since $C(\dd)$ is unchanged by removing any~1 argument, 
we have by definition $\theta_{\X, \thin n}=\theta_{\X-1, \thin n-1}$ for $n > \frac{3\X}5$), which implies that formula~\eqref{STRONG} still holds. Combining all three cases, we obtain the statement of Theorem~\ref{thm1}.
\end{proof}

Before proceeding, let us give an application 
of Theorem~\ref{thm1}.
Consider the {\it Painlev\'e I equation}:
\begin{align}\label{P1eq}
\frac{d^2 U}{dX^2} + \frac{1}{16}U^2-\frac{1}{16}X \=0\,.
\end{align}
This equation has (cf.~e.g.~\cite{JoshiK, Kapaev}) a unique formal solution $U(X)$
of the form
\begin{align}\label{solP1eq}
U(X) \= \sum_{g=0}^{\infty}c_g\, X^{\frac{1-5g}2}\,,  \qquad c_0=-1\,, \; c_g\in \mathbb{C}\,,
\end{align} 
where $c_g$, $g\ge0$, are determined by the recursion
\begin{align}\label{reccg}
c_g \= 50\,(g-1)^2\,c_{g-1} \+ \frac12\, \sum_{h=2}^{g-2}
c_h\, c_{g-h}\,, \quad g\ge3\,,
\end{align}
with the initial data $c_0=-1$, $c_1=2$, $c_2=98$. Here we use normalizations as in~\cite{DYZ0}.
From~\eqref{reccg} one can deduce the following asymptotics for $c_g$:
\begin{align}\label{asymcg}
c_g \;\sim\; A\cdot 50^g \, (g-1)!^2\, \Bigl(1-\frac{49}{3750 g^3}-
\frac{49}{1250 g^4}+\cdots\Bigr)\,, \quad g\to \infty\,,
\end{align}
where $A$ is some constant whose determination
requires deep analysis, which was obtained in~\cite{JoshiK, Kapaev} and we will give a new proof.

\smallskip

\noindent {\bf Theorem A (\cite{JoshiK, Kapaev}).}
 {\it  There holds that }
$$A\=\frac{1}{2\pi^2} \, \sqrt{\frac35}\,.$$

\begin{proof}  
According to~\cite{IZ} (cf.~also~e.g.~\cite{DYZ0}), we know that $c_g$ can be expressed in terms of Witten's intersection  
	numbers as follows:
	\begin{align}\label{cgrelation}
		c_g \= \frac{2^g\, 3^{3g-2}}{5^{3g-3}} \, \frac{(5g-5)!\, (5g-3)}{(3g-3)!}\, C(2^{3g-3})\,,
		\quad g\ge2\,.
	\end{align}
	The theorem is then proved by using Theorem~\ref{thm1} and Stirling's formula.
\end{proof}

Note that Theorem~A cannot be deduced from Aggarwal's result~\eqref{asyAgg}, as $n=3g-3$ which is certainly beyond $o(\sqrt{g})$.

We now proceed and prove Theorem~\ref{thm2}.

\begin{proof}[Proof of Theorem~\ref{thm2}]
For $Y\in\ZZ_{\ge1}$ and $\alpha\in\ZZ_{\ge0}$, define
\begin{align}
&\theta_{Y}^{(\alpha)} \: \sup_{n\ge1}\max_{\substack{\dd\in(\ZZ_{\ge1})^n \\ X(0^{\alpha},\dd)-p_1(\dd)=Y}}C(0^\alpha,\dd)\,, \label{defthetaYalpha}\\
&\phi_{Y}^{(\alpha)} \: \inf_{n\ge1}\min_{\substack{\dd\in(\ZZ_{\ge1})^n \\ X(0^{\alpha},\dd)-p_1(\dd)=Y}}C(0^\alpha,\dd)\,. 
\end{align}
From Theorem~\ref{thm1} and Lemma~\ref{lemmaC0C00} we know that as $Y\to\infty$, $\theta^{(\alpha)}_Y = 1/\pi + O(Y^{-1})$, $\alpha=0,1,2$, with an absolute $O$-constant. From Theorem~\ref{thm1}, inequalities~\eqref{C0dlowerbound} and~\eqref{C00dlowerbound} we find that as $Y\to\infty$, $\phi^{(\alpha)}_Y \= 1/\pi + O(Y^{-1})$, $\alpha=0,1,2$, with an absolute $O$-constant.

Consider $\alpha\ge3$. It suffices to consider two cases: 1. $Y\ge2\alpha-4$; 2. $\alpha-2\leq Y\leq 2\alpha-5$.

1. $Y\ge2\alpha-4$. For any $n\ge1$ and $\dd\in(\mathbb{Z}_{\ge1})^n$ satisfying $X(0^\alpha,\dd)=p_1(\dd)+Y$, write $\dd=(1^{p_1(\dd)},d_{p_1(\dd)+1},\dots,d_n)$. By using~\eqref{string} we have
\begin{align} C(0^\alpha,\dd)&\=C(0^\alpha,d_{p_1(\dd)+1},\dots,d_n)\nn\\
&\=\sum_{j=p_1(\dd)+1}^{n} \frac{2d_j+1}{3(Y-1)} C(0^{\alpha-1},d_{p_1(\dd)+1},\dots,d_j-1,\dots,d_n)\nn\\
&\;\leq\; \Bigl(1+\frac{3-\alpha}{3(Y-1)}\Bigr) \,  \max\{\theta^{(\alpha-1)}_{Y-1},\theta^{(\alpha-1)}_{Y-2}\}\,. \label{C0alphad}
\end{align}
So
\begin{align}
\theta^{(\alpha)}_{Y}\;&\leq\; \Bigl(1+\frac{3-\alpha}{3(Y-1)}\Bigr) \, 
\max\{\theta^{(\alpha-1)}_{Y-1},\theta^{(\alpha-1)}_{Y-2}\} \nn\\
\;&\leq\; \cdots \;\leq\; \, \max_{j=0,\dots,\alpha-2}\{\theta^{(2)}_{Y+2-\alpha-j}\}\,
\prod_{j=1}^{\alpha-2} \Bigl(1+\frac{2+j-\alpha}{3(Y-j)}\Bigr)\,. \label{thetagalphaest}
\end{align}
Given a sufficiently small $\epsilon>0$. When $\alpha\leq\epsilon Y$ we know from $\theta^{(2)}_Y \= 1/\pi + O(Y^{-1})$ that the maximum on the right-hand side of~\eqref{thetagalphaest} equals $1/\pi+O(Y^{-1})$ with an absolute $O$-constant; when $\alpha>\epsilon Y$, the maximum is still bounded by some constant but the product on the right-hand side of~\eqref{thetagalphaest} has the estimate
\begin{align}
\prod_{j=1}^{\alpha-2} \Bigl(1+\frac{2+j-\alpha}{3(Y-j)}\Bigr)\;\leq\;
\exp\biggl(\sum_{j=1}^{\alpha-2}\frac{2+j-\alpha}{3(Y-j)}\biggr) \;\leq\; \exp\biggl(-\frac{(\alpha-2)(\alpha-3)}{6Y}\biggr)\,, \label{upperbound}
\end{align}
and hence is $O(Y^{-1})$ with an absolute $O$-constant. Therefore, we conclude that
\begin{align}\label{thetaYalphaasy}
\theta_{Y}^{(\alpha)} \;\leq\; \frac{1}{\pi}\prod_{j=1}^{\alpha-2} \Bigl(1+\frac{2+j-\alpha}{3(Y-j)}\Bigr) +O(Y^{-1})\,,
\end{align}
with an absolute $O$-constant (independent of $\alpha$). 
Similarly,
\begin{align}
\phi^{(\alpha)}_{Y}\;&\geq\; \Bigl(1+\frac{3-\alpha}{3(Y-1)}\Bigr) \, 
\min\{\phi^{(\alpha-1)}_{Y-1},\phi^{(\alpha-1)}_{Y-2}\} \nn\\
\;&\geq\; \cdots \;\geq\; \, \min_{j=0,\dots,\alpha-2}\{\phi^{(2)}_{Y+2-\alpha-j}\}\,
\prod_{j=1}^{\alpha-2} \Bigl(1+\frac{2+j-\alpha}{3(Y+1-2j)}\Bigr)\,. \label{phigalphaest}
\end{align}
Similar to~\eqref{thetaYalphaasy} we can obtain
\begin{align}\label{phiYalphaasy}
\phi_{Y}^{(\alpha)} \;\geq\; \frac{1}{\pi}\prod_{j=1}^{\alpha-2} \Bigl(1+\frac{2+j-\alpha}{3(Y+1-2j)}\Bigr) +O(Y^{-1})\,,
\end{align}
with an absolute $O$-constant (independent of $\alpha$).

Set $$F_1(Y,\alpha)\=\prod_{j=1}^{\alpha-2}\Bigl(1+\frac{2+j-\alpha}{3(Y+1-2j)}\Bigr)\,,\quad
F_2(Y,\alpha)\=\prod_{j=1}^{\alpha-2} \Bigl(1+\frac{2+j-\alpha}{3(Y-j)}\Bigr)\,.$$
We have 
\begin{equation}\label{F1F2ineq}
F_1(Y,\alpha) \; \leq \; F_2(Y,\alpha) \;\leq\;
e^{-\frac{(\alpha-2)(\alpha-3)}{6Y}}\,.
\end{equation}
Note that $F_2(X(\dd)-p_1(\dd),p_0(\dd))$ appears in~\eqref{Cdasywithp0v2}. 
When $Y\ge 4\alpha+1$,
\begin{align}
&\log \frac{F_2(Y,\alpha)}{F_1(Y,\alpha)} \;\leq\; \sum_{j=1}^{\alpha-2} K_3\,\frac{\alpha^2}{Y^2}
\;\leq\; K_3\,\frac{\alpha^3}{Y^2}\,,
\end{align}
where $K_3>0$ is an absolute constant. So
\begin{align*}
|F_2(Y,\alpha)- F_1(Y,\alpha)|=F_2(Y,\alpha) \Bigl|\frac{F_2(Y,\alpha)-F_1(Y,\alpha)}{F_2(Y,\alpha)}\Bigr|
\leq  K_3\frac{\alpha^3}{Y^2}\exp\Bigl(-\frac{(\alpha-2)(\alpha-3)}{6Y}\Bigr)\,.
\end{align*}
It follows that $|F_2(Y,\alpha)- F_1(Y,\alpha)|\leq \frac{K_4}{\sqrt{Y}}$ for some absolute constant $K_4>0$. When $2\alpha-4\leq Y\leq 4\alpha$, by~\eqref{F1F2ineq} we also have $|F_2(Y,\alpha)- F_1(Y,\alpha)|\leq \frac{K_5}{\sqrt{Y}}$ for some absolute constant $K_5>0$. This proves~\eqref{Cdasywithp0v2}.

2. $\alpha - 2\leq Y\leq 2\alpha - 5$. Since for every $n\ge1$, $\dd\in(\ZZ_{\ge1})^n$, $X(0^{\alpha},\dd)-p_1(\dd)\ge \alpha-2$, we can improve inequality~\eqref{C0alphad} by
\begin{align}\label{C0alphadimproved}
\theta_{Y}^{(\alpha)} \;\leq\; (1+\frac{3-\alpha}{3(Y-1)}) \, \left\{
\begin{array}{ll}
\theta_{Y-1}^{(\alpha-1)}\,, \quad &{\rm if} \; Y\= \alpha-2\,,\\
\max\{\theta^{(\alpha-1)}_{Y-1},\theta^{(\alpha-1)}_{Y-2}\}\,, \quad &{\rm if} \; Y\geq \alpha-1\,.
\end{array}\right.
\end{align}
Iterating~\eqref{C0alphadimproved} $\alpha-2$ times, we obtain
\begin{align*}
\theta_{Y}^{(\alpha)} \;\leq\; \max_{j=0,\dots,Y+1-\alpha}\{\theta^{(2)}_{Y+2-\alpha-j}\}\,
\max_{\substack{S_0=0, S_{k}-S_{k-1}\in\{1,2\} \\ Y-S_k\ge \alpha-k-2\, (1\leq k\leq \alpha-2)}}\prod_{j=1}^{\alpha-2} \Bigl(1+\frac{2+j-\alpha}{3(Y-1-S_{j-1})}\Bigr)\,.
\end{align*} 
The first factor is bounded, and 
\begin{align*}
\prod_{j=1}^{\alpha-2} \Bigl(1+\frac{2+j-\alpha}{3(Y-1-S_{j-1})}\Bigr) \;\leq\; \prod_{j=1}^{\alpha-2} \Bigl(1+\frac{2+j-\alpha}{3Y}\Bigr) \;\leq\; \exp\Bigl(-\frac{(a-2)(a-3)}{6Y}\Bigr)\,,
\end{align*}
so $\theta_{Y}^{(\alpha)}\leq \frac{K_6}{\sqrt{Y}}$ for some $K_6>0$. By~\eqref{F1F2ineq}, we obtain~\eqref{Cdasywithp0v2}.
\end{proof}

We now use Theorem~\ref{thm2} to prove Corollary~\ref{cor1}.
\begin{proof}[Proof of Corollary~\ref{cor1}]
For $k=k(g)=O(\sqrt{g})$, we know from Theorem~\ref{thm2} and Stirling's formula that
\begin{align}\label{C0sqrtg}
C\bigl(0^{k},2^{3g-3+k}\bigr) \=\dfrac1\pi \, e^{-\frac{k^2}{30g}}\, (1+o(1)) \+ O\Bigl(\frac{1}{\sqrt g}\Bigr)\,.
\end{align}
For $k/\sqrt{g}\to\infty$, we have
\begin{align}
\prod_{j=1}^{k}\Bigl(1+\frac{2+j-k}{3(5g-5+2k-j)}\Bigr)\;\leq\; \Bigl(1+\frac{2}{3(5g-5+k)}\Bigr) \,
 \Bigl(1+\frac{1}{3(5g-4+k)}\Bigr) \nn\\
\times\, \exp \Bigl(-\frac{(k-2)(k-3)}{6(5g-5+2k)}\Bigr) \,\to\, 0\,, \quad g\to\infty\,.
\end{align}
By using Theorem~\ref{thm2} we know that $C(0^k,2^{3g-3+k})\to0$ 
as $g\to\infty$.
This proves Corollary~\ref{cor1}.
\end{proof}

\begin{remark}
Theorem~\ref{thm1} and equation~\eqref{string} also imply the following statement, 
 which is not completely uniform but has a better error term than Theorem~\ref{thm2}: 
{\it For $n\ge1$, $\dd\in  (\mathbb{Z}_{\ge0})^n$ satisfying $g(\dd)\in\mathbb{Z}_{\ge1}$ and for
	$p_0(\dd)$ fixed, we have
$$
C(\dd) \m \dfrac1\pi \prod_{j=1}^{p_0(\dd)} \Bigl(1+\frac{2+j-p_0(\dd)}{3X(\dd)-3p_1(\dd)-3j}\Bigr) \= O\Bigl(\thin\dfrac1 {g(\dd)}\Bigr)\,,
$$
uniformly as $g(\dd)\to\infty$.} 
The proof of this statement is given without any change by the one 
written in the first arXiv version of this paper of the Theorem~2 there.
\end{remark}

\section{Polynomiality in large genus}
\label{secproofofthm2}
In this section, we prove Theorem~\ref{thmpoly} by using the recursion~\eqref{CVirasoro}. 
\begin{proof}[Proof of Theorem~\ref{thmpoly}]
Fix $\dd'=(d_1,\dots,d_{n-1})\in(\mathbb{Z}_{\ge0})^{n-1}$. 
Write $\dd=(\dd',d_n)$ with $d_n\ge0$.
According to~\cite{LX} (cf.~\cite{GY}) and Stirling's formula we know that $C(\dd)$ has the asymptotic expansion
\begin{align}\label{Cdexpansion2}
C(\dd) \= \frac1\pi \,\sum_{k} \frac{C_k(\dd')}{X(\dd)^k}\,, \qquad \text{as } d_n \to \infty\,,
\end{align}
where $C_k$ are functions of~$\dd'$.
By using the recursion~\eqref{CVirasoro} and by performing Laurent expansions, we obtain 
\begin{align}
&C_k(d,\dd')-C_k(\dd') \= -\sum_{l=1}^{k-1}(-1)^{k-l}\binom{k-1}{l-1} \, C_{l}(d,\dd') +\frac{2}{3}\,\delta_{d,0}\,C_{k-1}(\dd')\nn\\
& + \frac{1}3\, \sum_{j=1}^{n-1}(2d_j+1)\, \bigl(C_{k-1}(d_1,\dots,d_j+d-1,\dots,d_{n-1})-C_{k-1}(\dd')\bigr) \nn\\
& + \sum_{\substack{a,b\\ a+b=d-2}}
\Biggl[\frac{2}3 \, \bigl(C_{k-1}(a,b,\dd')-C_{k-1}(\dd')\bigr) + \frac{1}3\,\sum_{I\sqcup J=\{1,\dots,n-1\}} \sum_{l=0}^{k-2} \, \mathfrak{a}_{(a,\dd'_I),l,k} \, 
C_l(b,\dd'_{J})  \Biggr]\,,  \label{Chatkrecursion}
\end{align}
where $d\ge0$ and $\mathfrak{a}_{\mathbf{w},l,k}$ are certain numbers defined via the following generating function
\begin{align}
(X(\mathbf{w})-1)! \, C(\mathbf{w})\, \frac{(X-X(\mathbf{w})-1)!}{X!\, (X-X(\mathbf{w}))^l} \;=:\;
\sum_k \frac{\mathfrak{a}_{\mathbf{w},l,k}}{X^k}\,.
\end{align}
Write
\begin{align}\label{Ck=ck}
C_k(\dd') \;=:\; \tilde{c}_k(p_0(\dd'),p_1(\dd'),\dots)\,,\quad k\ge0\,,
\end{align}
where $p_r(\dd')$ denotes the multiplicity of~$r$ in~$\dd'$, 
and $\tilde c_k(\mathbf{p})$ are functions of ${\bf p}=(p_0,p_1,\dots)$ with $p_i$ being integers satisfying $\sum_i p_i<\infty$.
Then formula~\eqref{Chatkrecursion} becomes
\begin{align}
& \tilde c_k(\mathbf{p}+\mathbf{e}_{d}) - \tilde c_k(\mathbf{p})\= -\sum_{l=1}^{k-1}(-1)^{k-l}\binom{k-1}{l-1} \, \tilde c_{l}(\mathbf{p}+\mathbf{e}_{d}) \nn\\
& +\frac{1}{3}\sum_{i\ge0}(2i+1)\,p_i \, \bigl(\tilde c_{k-1}(\mathbf{p}-\mathbf{e}_i+\mathbf{e}_{i+d-1})-\tilde c_{k-1}(\mathbf{p})\bigr) \nn\\
&  +\sum_{\substack{a,b\ge0\\ a+b=d-2}}
\Biggl[\frac{2}3 \,  \bigl(\tilde c_{k-1}(\mathbf{p}+\mathbf{e}_a+\mathbf{e}_b)-\tilde c_{k-1}(\mathbf{p})\bigr)\nn\\
&\quad + \frac{1}3\,\sum_{\substack{\mathcal{E}(\mathbf{t}+\mathbf{e}_a)\leq k-1 \\ 0\leq t_r\leq p_r, \,r\ge0}} \sum_{l=0}^{k-2} \, \biggl(
\tilde c_l(\mathbf{p}-\mathbf{t}+\mathbf{e}_b) \, \alpha_{\mathbf{t}+\mathbf{e}_a,l,k} \,\prod_{i\ge0}\binom{p_i}{t_i}\biggr)\Biggr]+\frac{2}{3}\,\delta_{d,0}\,\widetilde{c}_{k-1}(\mathbf{p})\,, \quad d\ge0\,,\label{hatCkrecursion}
\end{align}
where $\mathcal{E}(\mathbf{t}):=X(0^{t_0},1^{t_1},\dots)$, $\alpha_{\mathbf{t},l,k}=\mathfrak{a}_{(0^{t_0}1^{t_1}\cdots),l,k}$, and $\mathbf{e}_d$ denotes $(0,\dots,0,1,0,0,\dots)$ with ``1" appearing in the $(d+1)$th place. 
	
Let us now prove by induction that $\tilde c_k(\mathbf{p})$, $k\ge0$, belong to $\mathbb{Q}[p_0,p_1,\dots]$ and satisfy the degree estimates
\begin{equation}\label{degreeest1}
\deg \, \tilde c_k({\bf p}) \;\leq\; 3k-1\,,\quad k\ge1\,,
\end{equation} 
under the degree assignments $\deg p_d=2d+1$, $d\ge0$. For $k=0$, by using Theorem~\ref{thm1} 
we know that $\tilde c_0(\mathbf{p})\equiv1$. Assume that for $1\leq l\leq k-1$, $\tilde c_l({\bf p})\in\mathbb{Q}[p_0,p_1,\dots]$ are polynomials satisfying 
$\deg \, \tilde c_l({\bf p}) \leq 3l-1$. 
Then for $k$ and for every $d\ge0$, the RHS of equation~\eqref{hatCkrecursion} are polynomials in~$p_0,\dots,p_{[(3k-5)/2]}$.
Moreover, by the inductive assumption these polynomials are independent of~$d$ for every $d\ge 3k-1$, i.e.,
\begin{align}\label{diffck}
\tilde c_k(\mathbf{p}+\mathbf{e}_d) - \tilde c_k(\mathbf{p}) \= \left\{\begin{array}{ll}
f_{d}(p_0,\dots,p_{[(3k-5)/2]})\,, \quad &d\le 3k-2\,,\\
g(p_0,\dots,p_{[(3k-5)/2]})\,, \quad &d\ge 3k-1
\end{array}\right.
\end{align}
for some $f_d$ ($d\leq 3k-2$) and $g$ in $\mathbb{Q}[p_0,\dots,p_{[(3k-5)/2]}]$. The compatibility of~\eqref{diffck} implies that $g(p_0,\dots,p_{[(3k-5)/2]})\equiv A$ is a constant.  Solving~\eqref{diffck} we obtain that
$\tilde c_k$ have the form
\begin{align}\label{ckform1}
\tilde c_k(\mathbf{p}) \= h(p_0,\dots,p_{3k-2}) \+ A \, n'(\mathbf{p})\,,
\end{align} 
where $h\in \mathbb{Q}[p_0,\dots,p_{3k-2}]$ and $n'(\mathbf{p}):=\sum_{i\ge0}p_i$.
We aim to show that $A=0$. Consider equation~\eqref{hatCkrecursion} with $k$ replaced by $k+1$. Using a similar analysis and using~\eqref{ckform1}, we obtain that for every $d\ge 6\,k-1$,
\begin{align}\label{ck+1diff}
\tilde c_{k+1}(\mathbf{p}+\mathbf{e}_d) - \tilde c_{k+1}(\mathbf{p}) \= \frac{4}3\, A \, d \+ A'\, n'(\mathbf{p}) \+ s(p_0,\dots,p_{3k-2})\,,
\end{align} 
where $A'\in \mathbb{Q}$ is a constant, and $s(p_0,\dots,p_{3k-2})$ is a polynomial in $\mathbb{Q}[p_0,\dots,p_{3k-2}]$ independent of~$d$.  Now taking $\mathbf{p}=\mathbf{0}$, we find that equation~\eqref{ck+1diff} contradicts with~\eqref{twopointZograf} unless $A=0$. 
We conclude from~\eqref{ckform1} that
\begin{align}
\tilde c_k(\mathbf{p})\in\mathbb{Q}[p_0,\dots,p_{3k-2}]\,.
\end{align}
Now taking $d\ge 3k-1$ in equation~\eqref{hatCkrecursion} gives
\begin{align}
&0\=-\sum_{l=1}^{k-1}(-1)^{k-l}\binom{k-1}{l-1} \, \tilde c_{l}(\mathbf{p})  
\+\frac{1}{3}\sum_{i\ge0}(2i+1)\,p_i \, \bigl(\tilde c_{k-1}(\mathbf{p}-\mathbf{e}_i)-\tilde c_{k-1}(\mathbf{p})\bigr) \nn\\
& +  \sum_{a=0}^{[(3k-5)/2]}\Biggl(
\frac{4}{3}  \, \bigl(\tilde c_{k-1}(\mathbf{p}+\mathbf{e}_a)-\tilde c_{k-1}(\mathbf{p})\bigr) 
+ \frac{1}3\,\sum_{\substack{\mathcal{E}(\mathbf{t}+\mathbf{e}_a)\leq k-1 \\ 0\leq t_r\leq p_r, \, r\ge0}} \sum_{l=0}^{k-2} \, \biggl(
\tilde c_l(\mathbf{p}-\mathbf{t}) \, \alpha_{\mathbf{t}+\mathbf{e}_a,l,k} \,\prod_{i\ge0}\binom{p_i}{t_i}\biggr)\Biggr)\,.\label{hatCkidentity}
\end{align}
Using~\eqref{hatCkidentity} and~\eqref{hatCkrecursion}, we obtain
\begin{align}
&\Delta_{p_d} \tilde c_{k}(\mathbf{p})\= \m\sum_{l=1}^{k-1}(-1)^{k-l}\binom{k-1}{l-1} \, \Delta_{p_d} \tilde c_{l}(\mathbf{p}) 
\+\frac{1}3\,\sum_{i\ge0}(2i+1)\,p_i \, \Delta_{p_{i+d-1}} \tilde c_{k-1}(\mathbf{p}-\mathbf{e}_i) \nn\\
& \+ \frac{2}3\, \sum_{a=0}^{d-2}
\Delta_{p_a}\Delta_{p_{d-2-a}} \tilde c_{k-1}(\mathbf{p})
\, - \, \frac{4}3\,\sum_{a=d-1}^{[(3k-5)/2]}  \Delta_{p_a} \tilde c_{k-1}(\mathbf{p}) \nn\\
&\+ \frac{1}3\, \sum_{a=0}^{d-2}\sum_{\substack{\mathcal{E}(\mathbf{t}+\mathbf{e}_a)\leq k-1 \\ 0\leq t_r\leq p_r\,(r\ge0)}} \sum_{l=0}^{k-2} \biggl(
\, \Delta_{p_{d-2-a}} \tilde c_l(\mathbf{p}-\mathbf{t}) \, \alpha_{\mathbf{t}+\mathbf{e}_a,l,k} \,\prod_{i\ge0}\binom{p_i}{t_i}\biggr) \nn\\
& +\, \frac{1}3\,\sum_{a=d-1}^{[(3k-5)/2]} 
\sum_{\substack{\mathcal{E}(\mathbf{t}+\mathbf{e}_a)\leq k-1 \\ 0\leq t_r\leq p_r\,(r\ge0)}} 
\sum_{l=0}^{k-2} \, \biggl(\tilde c_l(\mathbf{p}-\mathbf{t}) \, \alpha_{\mathbf{t}+\mathbf{e}_a,l,k} \,\prod_{i\ge0}\binom{p_i}{t_i}\biggr)+\frac{2}{3}\,\delta_{d,0}\,\widetilde{c}_{k-1}(\mathbf{p})\,, \quad d\ge0 \,. \label{hatCk2}
\end{align}
We find by inductive assumption that each term of the right-hand side of~\eqref{hatCk2} is of degree less than or equal to $3k-2-2d$ for every $d\ge0$, which implies 
$\deg \tilde c_k\leq 3k-1$. In particular, $\tilde c_k$ is a polynomial that only depends on $p_0,\dots,p_{[3k/2]-1}$. 
	
Now restrict to the case when $\dd'\in(\mathbb{Z}_{\ge2})^{n-1}$.
From~\eqref{Cdexpansion2}, \eqref{degreeest} and Stirling's formula we know that 
\begin{align}
\widehat{C}(\dd) \;\sim\; \sum_k \frac{\widehat{C}_k(\dd')}{X(\dd)^k}\,, 
\qquad d_n\to\infty\,,\label{Chatexpansion}
\end{align}
where $\widehat{C}_k$ are functions of~$\dd'$ with $\widehat{C}_0\equiv1$.
Define $\widehat{c}_k(p_2,p_3,\dots)\in \mathbb{Q}[p_2,p_3,\dots]$, $k\ge0$, via
\begin{align}
\gamma(X)\,\sum_k \frac{\widehat{c}_k(p_2,p_3,\dots)}{X^k} \= \frac 1\pi\sum_k \frac{\tilde c_k(0,0,p_2,p_3,\dots)}{X^k}\,,
\end{align}
where the left-hand side is understood as a power series in~$X^{-1}$. 
It then follows from $\deg \tilde c_k\leq 3k-1$ that $\deg \widehat{c}_k\leq 3k-1$.
Using~\eqref{Chatexpansion},~\eqref{Cdexpansion2},~\eqref{Ck=ck}, we know that
\begin{align}
\widehat{C}_k(\dd') \=\widehat{c}_k(p_2(\dd'),p_3(\dd'),\dots)
\end{align}
for all~$\dd'$. The statement that $\widehat{c}_k(0,0,\dots)=0$ follows from the fact that $C(3g-2)=\gamma(2g-1)$ for every $g\ge1$. This finishes the proof of Theorem~\ref{thmpoly}.
\end{proof}

A few explicit expressions for $\tilde c_k$ and $\widehat{c}_k$ are given in 
Table~\ref{tablepoly}.  
\begin{table}[h!]
\begin{center}
\renewcommand\arraystretch{1.3}
\tabcolsep=0.6cm
\begin{tabular}{!{\vrule width 1.5pt}c|c|c!{\vrule width 1.5pt}}
			\hlinew{1.2pt}
			$k$ &  $\tilde{c}_{k}(0,0,p_2,p_3,\dots)$ & $\widehat{c}_{k}(p_2,p_3,\dots)$ \\ 
			\hlinew{1.2pt}
			0 & 1 & 1 \\
			\hline
			1 & $-\frac{17}{18}$ & 0 \\ 
			\hline
			2 & $\frac{613}{648}-\frac{5}{18}\,p_2$ & $-\frac{5}{18}\,p_2$ \\
			\hline
			3 & $\frac{535}{324}\,p_2+\frac{35}{27}\,p_3-\frac{33713}{34992}$ 
			& $\frac{25}{18}\,p_2+\frac{35}{27}\, p_3$ \\
			\hline
			\multirow{2}{*}{4}
			& $\frac{1225}{648}\,p_2^2-\frac{25025}{11664}\,p_2-\frac{3325}{1944}\,p_3$ & $\frac{1225}{648}\,p_2^2-\frac{185}{324}\,p_2
			-\frac{35}{72}\, p_3$ \\
			& $-\frac{1225}{216}\,p_4
			-\frac{385}{216}\,p_5+\frac{2424889}{2519424}$ &
			$-\frac{1225}{216}\,p_4-\frac{385}{216}\,p_5$ \\
			\hlinew{1.2pt}
		\end{tabular}
	\end{center}
	\caption{Expressions for $\tilde c_{k}(0,0,p_2,p_3,\dots)$, $\widehat{c}_{k}(p_2,p_3,\dots)$ with $k=0,\dots,4$.}
	\label{tablepoly}
\end{table}

\end{document}